\documentclass[journal, 12pt, onecolumn, draftclsnofoot]{IEEEtran}

\usepackage{graphicx}
\usepackage{amsmath,amssymb,amsthm}
\usepackage{algorithm,algorithmic}
\usepackage{enumerate}
\usepackage{xspace}
\usepackage{cite}
\usepackage{todonotes}
\usepackage{bbm}
\usepackage{subfig}
\usepackage{array}
\newcolumntype{L}[1]{>{\raggedright\let\newline\\\arraybackslash\hspace{0pt}}m{#1}}
\newcolumntype{C}[1]{>{\centering\let\newline\\\arraybackslash\hspace{0pt}}m{#1}}
\newcolumntype{R}[1]{>{\raggedleft\let\newline\\\arraybackslash\hspace{0pt}}m{#1}}

\newcommand{\cache}{\mathcal{I}}
\newcommand{\new}{\mathcal{N}}
\newcommand{\ind}[1]{\mathbb{I}_{\{#1\}}}

\newcommand{\E}{\mathbb{E}}
\newcommand{\EE}[1]{\E\left[#1\right]}
\newcommand{\Pchar}{\mathbb{P}}
\newcommand{\PP}[1]{\Pchar\left[ #1 \right]}
\newcommand{\EEcond}[2]{\E\left[\left.#1\right|#2\right]}
\newcommand{\PPcond}[2]{\Pchar\left[\left.#1\right|#2\right]}

\renewcommand{\S}{\mathcal{S}}
\newcommand{\A}{\mathcal{A}}
\newcommand{\G}{\mathcal{G}}
\newcommand{\N}{\mathbb{N}}
\newcommand{\R}{\mathbb{R}}
\newcommand{\rel}{\mathcal{O}}
\newcommand{\New}{\mathcal{N}}
\newcommand{\Value}{\mathcal{V}}

\newcommand{\Z}{\mathcal{Z}}
\newcommand{\paras}{{\theta}}
\newcommand{\para}{{\boldsymbol\paras}}
\newcommand{\thresh}{\mathcal{T}}

\newtheorem{theorem}{Theorem}
\newtheorem{lemma}{Lemma}
\newtheorem{definition}{Definition}

\newtheorem{corollary}{Corollary}
\newtheorem{remark}{Remark}
\DeclareMathOperator*{\argmin}{argmin}

\begin{document}
\title{A Reinforcement-Learning Approach to Proactive Caching in Wireless Networks\thanks{The authors are with the Department of Electrical and Electronic Engineering, Imperial College London, UK.
Email: \{samuel.somuyiwa12, a.gyorgy, d.gunduz\}@imperial.ac.uk.} \thanks{Parts of this work were presented at the Int'l Symp. on Modeling and Optim. in Mobile, Ad Hoc, and Wireless Nets. (WiOpt), and at the 2nd Content Caching and Delivery in Wireless Nets. Workshop (CCDWN), Paris, France, May 2017 \cite{SoGyGu17,SoGyGu17b}.}} 

\author{\IEEEauthorblockN{Samuel O. Somuyiwa, Andr\'as Gy\"orgy and Deniz G\"und\"uz}

}

\maketitle

\vspace{-.6in}
\begin{abstract}
We consider a mobile user accessing contents in a dynamic environment, 
where new contents are generated over time (by the user's contacts), and remain relevant to the users for random lifetimes. The user, equipped with a finite-capacity cache memory, randomly accesses the system, and requests all the relevant contents at the time of access. The system incurs an energy cost associated with the number of contents downloaded and the channel quality at that time. Assuming causal knowledge of the channel quality, the content profile, and the user-access behavior, we model the proactive caching problem as a Markov decision process with the goal of minimizing the long-term average energy cost. We first prove the optimality of a threshold-based proactive caching scheme, which dynamically caches or removes appropriate contents from the memory, prior to being requested by the user, depending on the channel state. The optimal threshold values depend on the system state, and hence, are computationally intractable. Therefore, we propose parametric representations for the threshold values, and use reinforcement-learning algorithms to find near-optimal parametrizations.
We demonstrate through simulations that the proposed schemes significantly outperform classical reactive downloading, and perform very close to a genie-aided lower bound. 
\begin{IEEEkeywords}
Proactive content caching, Markov decision process, reinforcement learning, policy gradient methods, wireless networks.
\vspace{-0.5cm}
\end{IEEEkeywords}
\end{abstract}

\IEEEpeerreviewmaketitle




\section{Introduction}
\label{sec:introduction}

Content delivery networks (CDNs), such as Amazon Web Service (AWS) and Akamai, replicate contents from a local repository at servers that are geographically closer to users; specifically, at Internet exchange points or Internet service providers. This approach significantly improves utilization of the Internet ``backbone'' capacity, thereby reducing latency and improving reliability \cite{CDN33}. However, today a large proportion of high-rate contents, e.g., videos, are delivered to users through cellular/wireless networks, which may introduce bottlenecks. 
Researchers have recently proposed \emph{proactive caching} of contents at the edge of the wireless network, that is, at the micro/macro base stations (BS) and/or even directly at user equipments (UEs) as a method to extend the idea of CDN all the way to the wireless network edge. 
Proactive caching is particularly appropriate for prerecorded contents, for example, YouTube videos or user generated contents in online social networks (OSNs), and is based on the assumption that the system knows/predicts in advance which contents are likely to be requested by the users. 

In \cite{FemtoCaching:INFOCOM, 6LivingontheEdge, Blasco:2014} proactive caching of contents at wireless access points is considered to reduce congestion in back-haul links. Caching at macro and micro BSs is modeled as a stochastic optimization problem in \cite{zhou_stoch_jrnl:2016}, with the objective of minimizing the average transmission power and delay in an heterogeneous wireless network. 
Proactive caching of contents directly at user devices has also been studied. Downlink energy efficiency of proactive caching is addressed in \cite{offlinecache:2015} and \cite{Gregori:JSAC:16}, which study the problem in the offline setting, that is, user demands and channel conditions are known in advance. Offline proactive caching is also studied in \cite{bitrateAdap} where, video bit-rate adaptation is used to improve the video streaming performance. Both offline and online proactive caching is considered in \cite{VincentPoor} to improve the effective throughput (hit-rate) given random user requests and a limited cache capacity. However, these works do not take into account the time-varying nature of contents, particularly in the context of OSNs, where content generation varies with time (and can also be bursty \cite{OSN_traffic:2016}), and the popularity of each content is non-stationary, such that content popularity typically diminishes soon after the content is generated \cite{Socialmedia_1:2013}. For example, the average \emph{lifetime}, that is, the period of time a content remains popular is approximately 2 hours for a video \emph{posted} on Facebook , and 18 minutes for a video \emph{tweeted} on Twitter \cite{M2Mconsulting:2016}. Instead, most caching schemes in the literature make caching decisions based on a static content-popularity profile and a fixed content library, which
results in performance degradation. For example, it is shown in \cite{Socialmedia_2:2014} that service delay increases, and cache-hit ratio decreases, when caching decisions for social media contents are done without taking the time variations in popularity into consideration. Time variations in the wireless channel quality and traffic conditions, together with variations in 
the lifetime and popularity of contents, require that proactive caching must incorporate an intelligent content placement and cache update mechanism.

In this paper, we consider proactive content caching into a mobile UE in the framework of an OSN, such that new contents (messages, videos, pictures), posted by the user's connections, become available over time. Each content has a lifetime, that is, the period of time it remains \emph{relevant} to the user. Contents are delivered to the user via wireless links, causing the serving BS to incur a transmission energy cost\footnote{Although we focus on the energy cost at the BS in this work, the proposed framework can be easily adapted to any other network resource, e.g., bandwidth, delay or the energy cost at the UE.}. This energy cost depends on the amount of content downloaded as well as the channel and network conditions, which typically vary over time due to traffic, user mobility, pathloss, as well as large scale fading effects. 

In \emph{reactive} content delivery, which is the conventional method employed in wireless networks today, every time the user accesses the OSN through an application software (\emph{app}) on his mobile device, all the contents that have not yet expired, that is, those that are still \textit{relevant}, are downloaded to the UE. The energy cost of reactive delivery depends on the channel conditions at the time of user access and the number of contents downloaded. Alternatively, \textit{proactive caching} allows downloading contents to the UE before the user accesses the OSN to request these contents. Downloaded contents are stored in the cache, and are retrieved and delivered to the application layer, i.e., to the app, whenever the user accesses the OSN. Therefore, contents can be downloaded under more favorable channel conditions, leading to reduced energy consumption. On the other hand, proactive caching may push contents that will not be requested by the user within their lifetimes, creating unnecessary energy consumption. Moreover, the limited cache capacity at the UE limits the amount of contents that can be proactively cached, or may require replacing already downloaded contents, increasing the cost. Hence, we aim to answer the question of which contents, and at what time, should be pushed to the cache. 

We consider a slotted time model, in which a random number of relevant contents are generated with random lifetimes at each time slot. 
For simplicity, we assume that all contents have equal size, which is without loss of generality if we assume that larger contents are split into smaller chunks of equal size e.g., video segments in DASH. We model both the channel quality and the user behavior as stochastic processes. The user randomly accesses the contents in her OSN feed in order to view/consume all the relevant contents at the time of access. We will initially assume that the statistics of the underlying stochastic processes are known, but we will later propose reinforcement-learning algorithms that can learn and adapt to an unknown environment. 

Our contributions are summarized as follows.
\begin{itemize}
\item We formulate the problem as an infinite-horizon average-cost Markov decision process (MDP) \cite{Bertsekas:book}, with the objective of
minimizing the long term average energy consumption.  

\item To overcome the technical difficulty due to the continuous distribution of the channel quality, we introduce a new MDP model, referred to as an MDP with side information (MDP-SI). 
We show the optimality of a threshold-based proactive caching policy, which downloads contents into, and removes contents from, the cache
depending on the remaining lifetime of the contents and the relative value of the current channel state with respect to a threshold. 

\item Since the optimal threshold values depend on the system state, the prohibitively large size of the state space makes it practically infeasible to compute and store them. Hence, we introduce two low-complexity parametric policy representations that are able to approximate the optimal performance. 
The first policy, called \emph{Longest lifetime In--Shortest lifetime Out (LISO)}, assigns a single threshold to each pair of contents with the longest remaining lifetime outside the cache, and content with the shortest remaining lifetime inside the cache, independent of the system state. 
The second policy, called linear function approximation (LFA), represents the threshold values for every possible pair of remaining lifetimes as a linear function of the system state. 


\item We use reinforcement learning techniques obtain the optimal threshold values for the proposed caching schemes. In particular, we apply two policy gradient schemes,  he finite difference method (FDM) and the likelihood-ratio method (LRM) \cite{Policysearch}.  

\item To evaluate the performances of the proposed caching schemes, we describe two lower bounds: one assuming unlimited cache capacity, and another assuming non-causal knowledge of the user-access times. Through numerical simulations, we demonstrate that the proposed schemes perform close to the latter lower bound when the cache capacity is small, and to the former for larger cache capacities, and significantly outperform reactive caching.
We also show that the LFA policy outperforms LISO to some considerable extent, 
and that the PG with LRM finds a better solution than with FDM.

\item Finally, we introduce memory into the stochastic processes for content (lifetime) generation and channel quality, and show via simulations that the performance gain of the LFA policy over LISO is more significant in such scenarios.
\end{itemize}

\section{System Model}
\label{sec:system_model}

We consider a slotted (discrete time) communication system. At the beginning of each time slot $t$, a random number of contents, denoted by $M_t$, are generated. We denote the set of newly generated contents by $\New_t$, where $|\New_t|=M_t$. Each content is generated with a lifetime, after which it becomes irrelevant. In particular, if, at the beginning of time slot $t$, content $i$ is generated with lifetime $K_{t,i}$, it can be consumed in time slots $t, t+1,\ldots,t+K_{t,i}-1$, and otherwise will be removed from the system after time slot $t+K_{t,i}-1$. We denote the set of contents that are already in the cache at the beginning of time slot $t$ by $\cache_{t}$, and the set of relevant contents not inside the cache, including the $M_t$ newly generated contents, by $\rel_{t}$. The system architecture is illustrated in Fig. \ref{fig:sys_architecture}.

\begin{figure}[t]
\centering
\includegraphics[scale=0.4]{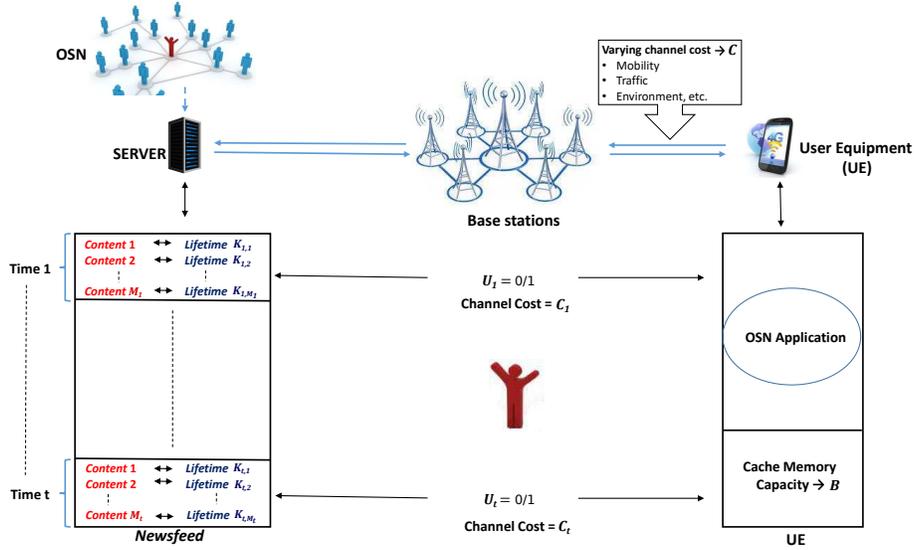}
\caption{Illustration of the system architecture. The OSN server has a list of relevant contents with random lifetimes for the user. Contents can be pushed to the UE cache memory before being requested by the user, to take advantage of favorable channel conditions.}
\label{fig:sys_architecture}
\vspace{-0.5cm}
\end{figure}

At each time slot, the user either accesses the system and consumes all the relevant contents, or does not access the system. The user access behavior is represented by the binary random variable $U_t$; that is, $U_t=1$ if the user accesses the system, and $U_t=0$ otherwise. When the user accesses the system, all the contents that are not already in the cache, $\rel_{t}$, are downloaded, and moved, together with all the contents in the cache, that is, $\cache_t$, to the app. If $U_t=0$, the cache manager (CM)
has the option of downloading some contents to, and removing others from the cache. We denote the set of contents that are downloaded at time slot $t$ by $A^{(1)}_t \subset \rel_{t}$, and the set of contents that are discarded from the cache by $A^{(2)}_t \subset \cache_{t}$. To unify notation, if $U_t=1$ we set $A^{(1)}_t=\rel_{t}$ and $A_t^{(2)} = \cache_{t}$.\footnote{We do not allow the CM to download and remove the same content in the same time slot, which is obviously suboptimal.}

Throughout the paper, since all the contents have the same size, it will be convenient to represent each content by its remaining lifetime. 
Following this representation, all sets of contents, that is, $\New_t$, $\rel_t$, $\cache_t$, $A^{(1)}_t$ and $A^{(2)}_t$, are multisets of remaining lifetimes (positive integers, with the set of all positive integer tuples denoted by $\N^*$). To simplify the treatment, when it does not cause confusion, we will only talk about sets instead of multisets, or subsets instead of sub-multisets of multisets, and operations, such as union, should be treated in a multiset manner. For a multiset $Z$ with positive elements, we let $Z-1=\{z > 0: z+1 \in Z\}$ denote the multiset obtained by reducing each element of $Z$ by $1$ and removing the elements which become $0$. With these definitions in mind, if $U_t=0$, the system evolves according to the following equations:
\begin{equation}
\begin{split}
\cache_{t+1} &= \big(\cache_t \cup A^{(1)}_t \setminus A^{(2)}_t\big) -1, \\
\rel_{t+1} &= \bigg(\big(\rel_t \cup A^{(2)}_t \setminus A^{(1)}_t\big) -1\bigg) \cup \New_{t+1},
\label{eq:set_noaccess}
\end{split}
\end{equation}
and according to the following equations if $U_t=1$:
\begin{align}
\cache_{t+1} = \emptyset \quad \mbox{and} \quad \rel_{t+1} &= \New_{t+1}.
\label{eq:set_access}
\end{align}


We assume that the user is equipped with a cache of capacity $B$, that is, $|\cache_t| \le B$, $\forall t$. Hence, the CM's actions, $A_t=(A^{(1)}_t,A^{(2)}_t)$, are constrained by the available cache capacity, and any valid action leads to a new state with $|\cache_t|\le B$.

Downloading a content at time $t$ has a cost $C_t$ that depends on the channel state. The total instantaneous cost at time $t$ is
$\mu_t = |A^{(1)}_t| \cdot  C_t$,
while the average cost after $T$ time slots is given by
$J_T = \frac{1}{T}\sum_{t=1}^{T} \mu_t$.
The goal is to minimize the long-term expected average cost defined as
\begin{equation*}
\rho \triangleq \limsup_{T\to\infty} \EE{J_T} = \limsup_{T\to\infty} \EE{\frac{1}{T}\sum_{t=1}^{T} \mu_t}.
\end{equation*}

\subsubsection{User Access Model}
We assume that the user access sequence $\{U_t\}$ is an arrival process with i.i.d. inter-arrival times $\{D_n\}$, where $D_n$ denotes a positive-integer-valued random variable. 
Throughout, we will make one of the two assumptions regarding $D_n$: (i) Bounded inter-arrival times, i.e., $D_n$ is bounded as $0\le D_n \le D_{max}$, where $D_{max} \in \mathbb{Z}^+$. (ii) Geometric inter-arrival times: $D_n$ has a geometric distribution with parameter $p_a$; hence $\{U_t\}$ is an i.i.d. process with $\PP{U_t = 1} = p_a$ (in turn, $D_{\max}=\infty$ in this case). The latter assumption, which is standard in the literature, is also known as the \emph{independent reference model} (IRM) (see, e.g., \cite{FemtoCaching:INFOCOM,zhou_stoch_jrnl:2016}).

\subsubsection{Content Generation Model}
We assume that $\{M_t\}$ is an i.i.d. sequence with generic random variable $M$, and is upper-bounded by $M_{max}\in \mathbb{Z}^+$. We further assume that the lifetimes are also i.i.d. with generic random variable $K$, and upper-bounded by $K_{\max}\in \mathbb{Z}^+$. 

\subsubsection{Channel Model}
We assume that the energy cost for downloading a content $C_t > 0$ is a continuous random variable with cumulative distribution function (cdf) $F_C(c)$, and it is assumed to be i.i.d. across time, and bounded by $C_{max}\in \R^+$. Aside from simplifying our system model, the i.i.d. assumption here is appropriate for micro BS deployments, where the user switches micro BSs across time slots.
We assume that the micro BSs can operate at the same time without any interference because they operate at a relatively low transmit power.
We also assume zero download delay \cite{zhou_stoch_jrnl:2016,VincentPoor}, implying that the duration of a time slot is long enough to download the required contents. Hence, the channel is approximately ergodic within a time slot, and is only subject to large-scale fading effects.

In the rest of the paper, we assume that the sequences $\{C_t\},\{D_n\},\{M_t\}$,$\{K_{t,i}\}$ are independent of each other.
In the following section, we assume that the CM is aware of the above stochastic model governing the system behavior.

\section{Optimal solution}
\label{sec:optimal}

In this section we derive a general result concerning the structure of the optimal cache management policy. First we define a special class of MDPs, which we call MDPs with side information (MDP-SI), and show that our cache management problem is an instance of this class. Then, we derive a general structural result for optimal policies in MDP-SI under some assumptions, and show that they apply to our problem.

\subsection{Standard MDP model}
\label{sec:MDP}
A finite-state finite-action MDP is characterized by a quadruple $(\S,\A,P,\mu)$, where $\S$ and $\A$, the state and action spaces, respectively, are finite sets, $P:\S \times \A \times \S \to [0,1]$ is a probability kernel (we will write $P(s'|s,a)$ instead of $P(s,a,s')$), and $\mu:\S \times \A \to [0,\mu_{max}]\cup\{\infty\}$ is a cost function with some $\mu_{max}>0$. 
The purpose of introducing an infinite cost is to allow a different action set in every state without complicating the notation too much: for every state $s \in \S$, the set $\A_s=\{a \in \A: \mu(s,a)<\infty\}$ denotes the set of feasible actions (otherwise the agent suffers infinite cost), and we assume that $\A_s \neq \emptyset$, $\forall s \in \S$. 
In an MDP, an agent controls a Markov chain and pays some cost over time. Assuming the agent selects an action $a \in \A_s$ at state $s \in \S$, the system evolves to state $s'$ with probability $P(s'|s,a) \triangleq \PPcond{S_{t+1} = s'}{S_t = s, A_t = a}$, where $\sum_{s'\in\S}P(s'|s,a)=1, \forall s\in\S,a\in\A$. The cost of taking action $a$ in state $s$ is $\mu(s,a)$. Denoting the state of the system at time $t$ by $S_t$ and the agent's action by $A_t$, the agent's goal is to minimize the infinite horizon average cost $\rho=\lim_{T\to\infty}\EE{\frac{1}{T}\sum_{t=1}^T \mu(S_t,A_t)}$.

A deterministic policy is a mapping $\pi:\S \to \A$, which selects a single action for each state; and let $\Pi$ denote the set of all deterministic policies.
For any policy $\pi$ let $P^\pi:\S\times\S \to [0,1]$ denote the transition kernel induced by $\pi$, that is $P^\pi(s'|s)=P(s'|s,\pi(s))$. Assuming the Markov chain defined by $P^\pi$ is irreducible and aperiodic for all $\pi$, let $\rho^\pi$ denote the infinite-horizon average cost $\rho$ when $A_t=\pi(S_t)$, that is,
\begin{equation}
\rho^{\pi} = \lim_{T\to\infty}\EE{\frac{1}{T}\sum_{t=1}^T \mu(S_t,\pi(S_t))}.
\label{eq:rho_pi}
\end{equation}
Due to our assumption on $P^\pi$, the initial state $S_0$ does not matter, and the limit in \eqref{eq:rho_pi} exits thanks to the non-negativity assumption on $\mu(s,a)$. It is well-known (see, e.g., \cite{Puterman:book}) that there exists a deterministic policy $\pi^*$ that minimizes the infinite-horizon average cost over all, possibly non-stationary and non-deterministic causal control policies, that is,
\begin{equation}
\pi^* = \argmin_{\pi} \rho^{\pi},
\label{eq:arg_opt_policy}
\end{equation}
where the minimum is taken over all admissible (causal) control strategies of the agent, in which $A_t$ may depend on the history $H_t \triangleq (S_1,\ldots,S_t,A_1,\ldots,A_{t-1})$ and some randomization.

\subsection{MDPs with side information (MDP-SI)}
\label{sec:MDPSI}
In the MDP-SI model, we extend the classical MDPs such that there is an i.i.d. sequence of side information $Z_t \in \Z$ for some $\Z \subset \R$, which is available to the agent before selecting $A_t$, and effects the cost $\mu$, that is, $\mu: \S\times\A\times\Z \to [0,\mu_{max}]\cup\{\infty\}$. Then the decision of the agent may depend on $H_t$, the randomization, and $(Z_1,\ldots,Z_t)$. This setup can be easily modeled in the MDP framework by changing the state space to $\S \times \Z$, but if $\Z$ is not finite, the analysis of the resulting MDP is significantly more complicated. Before delving into the analysis of the MDP-SI model, first we show that our problem can be cast as an MDP-SI problem.

At the end of time slot $t$, the state of the contents can be described by the sets $\cache_t$ and $\rel_t$, while the state of the user can be described by the time elapsed since the last access, denoted by $E_t$. To be precise, we assume that the user accesses the OSN at time $t=0$ (i.e., we set $U_0=1$); then $E_t$ is defined as $E_t \triangleq \min\{t-n: t>0, 0\le n \le t, U_n=1\}$. We denote by $\S \subset \N^* \times \N^* \times \N$ the set of all possible combinations of $\rel_t, \cache_t$, and $E_t$. Under the IRM user-access model (i.e., when the user-access process is i.i.d. and the inter-access times are geometrically distributed), the memoryless property of the geometric distribution implies that the exact value of $E_t$ does not affect the future given that $E_t>0$, and hence, the state of the user can be redefined as $\ind{E_t>0}$, the indicator function of the event $\{E_t>0\}$.%
\footnote{For an event $\mathcal{E}$, $\ind{\mathcal{E}}=1$ if $\mathcal{E}$ holds, and $0$ otherwise.} 
Unless otherwise stated explicitly, we will use $\ind{E_t}$ in place of $E_t$ for the IRM model; accordingly, $\S$ will denote the possible combinations of $\rel_t, \cache_t$, and $\ind{E_t>0}$. Note that under both of our user-access models (i.e., IRM or bounded inter-access times--$D_{\max}<\infty$), the state space $\S$ is finite. Furthermore, let $\A_s$ denote the set of download/discard actions available to the CM in a state $s\in \S$. The action of the agent in time slot $t$ is the pair $A_t=(A^{(1)}_t,A^{(2)}_t)$, and $C_t$ can be regarded as the i.i.d. side information $Z_t$. Indeed, the decision of the CM (i.e., the agent) depends on $C_t$, as the cost of action $A_t$ is $\mu(S_t,A_t,C_t)=C_t \cdot |A^{(1)}_t|$. 

The state $s \in \S$ of the system evolves according to \eqref{eq:set_noaccess} and \eqref{eq:set_access}, where the user access sequence depends on $E$, which evolves independently according to the distribution of $D_n$. The channel cost $C_t$, which is the side information, also evolves independently, with cdf $F_C$ in every time slot $t$. These independence assumptions ensure that the resulting model is indeed an MDP-SI.

\subsection{Structure of the optimal policy in MDP-SI}

In this section we derive the structure of the optimal policy for a general MDP-SI under certain conditions. To begin with, assume we have an MDP-SI characterized by $(\S,\A_{SI},P_{SI},\mu_{SI},\Z,F_Z)$, where $F_Z$ is the cdf of the real-valued side information, $\S$ and $\A_{SI}$ are countable, and $\Z \subset \R$. Let $\A$ denote the set of Borel-measurable\footnote{While throughout the paper we assume the existence of the necessary probability spaces and the measurability of functions as required, here we emphasize Borel measurability since it will be used explicitly in some proofs.} functions $\{g:\Z \to \A_{SI}\}$, and consider the MDP $(\S,\A,P,\mu)$ where 
\begin{align}
P(s'|s,g) = \EE{P_{SI}(s'|s, g(Z))},
\label{eq:MDPtoMDP-SI_prob}
\end{align}
and 
\begin{align}
\mu(s,g)=\EE{\mu_{SI}(s,g(Z),Z)},
\label{eq:MDPtoMDP-SI_cost}
\end{align}
where the expectations are taken over $F_z$. It is easy to see that any deterministic policy $\pi_{SI}:\S\times\Z \to \A_{SI}$ for the MDP-SI can be turned into a deterministic policy for the corresponding MDP using
\begin{equation}
\pi(s)=\pi_{SI}(s,\cdot)\in \A,
\label{eq:MDPtoMDP-SI_policy}
\end{equation}
and vice versa, and that the expected average cost of the two models are the same for the corresponding policies. Therefore, it is enough to consider the MDP $(\S,\A,P,\mu)$. If $\Z$ is finite, the new MDP is finite, and we can use standard results (see, e.g., \cite{Puterman:book}) to analyze the structure of the optimal policy:
Assume that the MDP is finite, $P^\pi$ is irreducible and aperiodic for any deterministic policy $\pi \in \Pi$, and let $S_1,A_1,S_2,A_2,\ldots$ denote the state-action sequence obtained by following policy $\pi$. Then $\rho^\pi$ in \eqref{eq:rho_pi}
exists, and the differential value function for any state $s \in \S $ is defined as
\begin{equation}
V^\pi(s) = \EEcond{\sum_{t=1}^\infty (\mu(S_t,A_t)-\rho^\pi)}{S_1=s}~.
\label{eq:diff_valuefunc}
\end{equation}
Furthermore, the optimal policy $\pi^*$ in \eqref{eq:arg_opt_policy} satisfies
\begin{equation}
\label{eq:opt_diff_valuefunc}
V^{\pi^*}(s) = \min_{a\in\A} \left\{\mu(s,a) - \rho^{\pi*} + \sum_{s' \in \S} P(s'|s,a) V^{\pi^*}(s') \right\},
\end{equation}
and $a=\pi^*(s)$ minimizes the right hand side. While these results make the analysis easy, unfortunately they do not directly apply to our case, where the state space $\S$ can be countably infinite and, due to the fact that $\Z$ is not finite, the action set $\A$ is infinite (and uncountable). Luckily, it is possible to extend the above results, specifically \eqref{eq:opt_diff_valuefunc} to  the MDP $(\S,\A,P,\mu)$ when $\Z$ is an interval (this will be done in the proof of Lemma~\ref{lem:piecewise_constant}).
Using the definition of the MDP in \eqref{eq:MDPtoMDP-SI_prob}--\eqref{eq:MDPtoMDP-SI_policy} and the expression for the optimal value function in \eqref{eq:opt_diff_valuefunc}, we can prove the following property of the optimal policy $\pi^*(s,\cdot)$ of an MDP-SI (the proof is given in Appendix~\ref{app:proof-piecewise_constant}). 
\begin{lemma}
\label{lem:piecewise_constant}
Consider the countable-state, finite-action MDP-SI problem $(\S,\A_{SI},P_{SI},\mu_{SI},\Z,F_Z)$. Suppose that $P^\pi$ is ergodic for any policy $\pi$ in the corresponding MDP $(\S,\A,P,\mu)$ and that $\Z$ is an interval.
Then \eqref{eq:opt_diff_valuefunc} holds for $(\S,\A,P,\mu)$. Furthermore, if $\mu_{SI}(s,a_{SI},z)$ is a linear function of $z$ 
for any $s\in\S, a_{SI}\in\A_{SI}$,
then the optimal policy $\pi^*(s,\cdot)$ is a piecewise constant function for any $s \in \S$.
\end{lemma}

Combining the MDP-SI formulation of Section~\ref{sec:MDPSI} with Lemma~\ref{lem:piecewise_constant}, we obtain that for any state $s\in\S$, the optimal decision is a piecewise constant function of the channel cost $C_t$ with values taken from $\A_s$. Also note that the technical condition in the lemma that $P^\pi$ is ergodic is easily satisfied in our cache management problem due to the user access model: The fact that the user access clears all contents from both inside and outside of the cache at least once in every $D_{\max}$ time slots or after a geometric waiting time ensures that any state visited with positive probability is positive recurrent. Thus, to achieve ergodicity, it is sufficient to guarantee that the process is aperiodic. This readily follows from the IRM model, and also holds for the bounded inter-access time model under mild assumptions (e.g., if the user can access the OSN in two consecutive time slots with positive probability).

\subsection{Structure of the optimal cache management policy}\label{sub:OptimalPolicy}

Here we will describe the structure of the optimal policy for the proactive caching problem. 

We start with the technical definition of partial ordering for multisets, which will be useful to characterize the effect of good actions: For two multisets $M_1$ and $M_2$ with nonnegative elements, we write $M_1 \le M_2$, if, either (i) they are of equal size and there is a one-to-one mapping between the elements of $M_1$ and $M_2$ such that the element from $M_1$ is never larger than the corresponding element from $M_2$; or (ii) if they are of different size, but the same relationship holds after adding zeros to the smaller set to equalize their sizes.

Now consider two states of the MDP describing the caching problem: $s=(\rel,\cache,E) \in \S$ and $s'=(\rel',\cache',E') \in \S$. We will say that $s$ \textit{is better than} $s'$, and write $s \succeq s'$, if $E=E'$,  the remaining lifetimes of all the contents are the same, that is, $\rel \cup \cache = \rel' \cup \cache'$, and $\rel \le \rel'$ and $\cache \ge \cache'$. Intuitively, $s \succeq s'$ means that the same contents are available for pre-caching in $s$ and $s'$, but in state $s$ ``better'' contents have already been downloaded to the cache (i.e., the contents in the cache remain relevant longer while the ones outside expire earlier). The next lemma formalizes this statement:

\begin{lemma}
\label{lem:Sbetter}
Assume the conditions of Lemma~\ref{lem:piecewise_constant} hold. Let $s, s' \in \S$ and suppose that $s \succeq s'$. Then, $V^{\pi^*}(s) \le V^{\pi^*}(s')$, that is, the future average download cost starting form $s$ is not larger than the cost starting from $s'$.
\end{lemma}
\begin{proof}
It is easy to see that if any action $a'$ is performed in $s'$, it is always possible to find another action $\hat{a}$ in $s$ such that the cost of $\hat{a}$ is no more than that of $a'$, that is, $\mu(s',a') \ge \mu(s,\hat{a})$, and the resulting new states satisfy $\hat{s}_2 \succeq s'_2$, where $s_2'$ and $\hat{s}_2$ denote the next state for the chains starting from $s'$ and $s$, respectively, assuming the content generation process and the user access process are the same (e.g., if $a'$ downloads a content from outside the cache of $s'$, $\hat{a}$ should download the content with the largest remaining lifetime from outside the cache of $s$, unless all the contents in the cache of $s$ have larger lifetimes, in which case $\hat{a}$ should not do anything). Now consider three coupled realizations of the MDP: $\{(S'_t,A'_t)\}$ starts from $S'_1=s'$, and follows the optimal policy $\pi^*$; the second realization $\{(\hat{S}_t,\hat{A}_t)\}$ starts from $\hat{S}_1=s$, and selects $\hat{A}_t$ such that $\hat{S}_t \succeq S'_t$ and $\mu(\hat{S}_t,\hat{A}_t) \le \mu(S'_t,A'_t)$ for all $t$; finally, $\{(S_t,A_t)\}$ starts form $S_1=s$, and follows the optimal policy $\pi^*$. Then, using the optimality of $A_t$ and $\pi^*$, by \eqref{eq:opt_diff_valuefunc} (which holds by Lemma~\ref{lem:piecewise_constant}), we have
\begin{align*}
V^{\pi^*}(s) = V^{\pi^*}(S_1) &\le \mu(\hat{S}_1,\hat{A}_1) - \rho^{\pi^*} + \EE{V^{\pi^*}(\hat{S}_2)}\\
& \le \mu(\hat{S}_1,\hat{A}_1) - \rho^{\pi^*} + \EE{\mu(\hat{S}_2,\hat{A}_2) - \rho^{\pi^*}
+\EE{V^{\pi^*}(\hat{S}_3)}}  \\
& \quad \vdots \\
& \le \EEcond{\sum_{t=1}^\infty (\mu(\hat{S}_t,\hat{A}_t) - \rho^{\pi^*})}{\hat{S_1}=s}~.
\end{align*}
Furthermore, by the coupling of the realizations, 
\[
\EEcond{\sum_{t=1}^\infty (\mu(\hat{S}_t,\hat{A}_t) - \rho^{\pi^*})}{\hat{S_1}=s}
 \le
 \EEcond{\sum_{t=1}^\infty (\mu(S'_t,A'_t) - \rho^{\pi^*})}{S'_1=s'} = V^{\pi^*}(s')~.
\]
Putting everything together, we obtain $V^{\pi^*}(s) \le V^{\pi^*}(s')$.
\end{proof}

\bigskip

Next we express the actions in $\A_s$ more intuitively by defining a \emph{simple action}, which we also denote as $a$, to simplify the notation.

\begin{definition}[\textbf{Simple Action}]
For any $l \in \cache$ and $L \in \rel$ (recall that $l$ and $L$ denote the remaining lifetime of some contents), a simple action $a=(l|L)$ is defined as follows: If $E_t >0$, $a = (l|L)$ replaces a cache content with remaining lifetime $l$ with a relevant content outside the cache with remaining lifetime $L$, by removing the former content from the cache and downloading and caching the latter; i.e., it ``swaps'' the two contents. If $E_t = 0$, $a = (l|L)$ downloads the content with remaining lifetime $L$, and moves both contents to the app.
\end{definition}

In this definition, $l=0$ means that the content with lifetime $L$ is downloaded without any corresponding removal of a content from the cache. Similarly, $L=0$ means that no content is downloaded while a content with remaining lifetime $l$ is removed from the cache. Note that, with the optimal policy, the latter (i.e., $L=0$) can only happen if either $l=0$ (i.e., an expired content is removed from the cache), or when $E_t = 0$ and more contents are moved from the cache to the app than those downloaded from the OSN server to the app. At every time slot $t$, because of the cache capacity constraint, the CM can only take up to $B$ simple actions if $E_t>0$. Therefore, at such instances, any action of an optimal policy can be expressed as at most $B$ consecutive simple actions, and an action $A_t = (\{L_1,\ldots,L_{B'}\},\{l_1,\ldots,l_{B'}\})$,\footnote{If either $|\rel|$ or $|\cache|$ is less than $B'$, we simply zero-pad the set so that $\vert A^{(1)}_t\vert = \vert A^{(2)}_t\vert = B'$. } for some $0 \le b \le B$ can be written as a sequence of simple actions $\{(l_1|L_1)\cdots(l_{b}|L_{b})\}$. 

Now for a state $s=(\rel,\cache,E)$ with $E>0$, assume that $l_1\le \cdots \le l_B$ are the contents in $\cache$, and let $L_1 \ge \cdots \ge L_B$ denote the $B$ largest elements of $\rel$.
To find the optimal action, first we determine the best simple action. Let $s^*_1$ denote the next state if action $(l_1,L_1)$ is taken, and let $s'_1$ denote the state after a different simple action $(l'|L')$. Since $l_1$ is the smallest element of the cache and $L_1$ is the largest element outside the cache, assuming the same content generation, it is immediate that $s^*_1 \succeq s'_1$.
Then, by Lemma~\ref{lem:Sbetter}, $V^{\pi^*}(s^*_1) \ge V(\pi^*)(s'_1)$. Therefore, by \eqref{eq:opt_diff_valuefunc}, $(l_1,L_1)$  is the best simple action. Considering larger actions composed of $b$ simple actions for $b \ge 2$, it follows similarly that the optimal such action is $A^{b}=\{(l_1|L_1)\cdots(l_{b}|L_{b})\}$ (note that the energy cost associated with such an action is $b\, C$, where $C$ is the channel cost of a single download). To find the optimal action, it remains to compare the actions $A^{b}$ for different values of $b \in \{0,\ldots,B\}$. Denoting the next state following action $A^{b}$ by $s^{b}$, the relative action value of $A^{b}$ for channel cost $C$ is given by
\begin{equation*}
Q^{\pi*}(s,A^{b},C)= b\, C - \rho^{\pi^*} + \EE{V^{\pi^*}(s^{b})},
\end{equation*}
and, by \eqref{eq:opt_diff_valuefunc}, the optimal action is the one minimizing $Q^{\pi*}(s,A^{b},C)$ for a given $C$: that is, $A^{b^*}$ with $b^*=\argmin_{b} Q^{\pi*}(s,A^{b},C)$.
Notice that, as a function of $C$, $Q^{\pi*}(s,A^{b},C)$ is a linear function with slope $b$ and intersecting the $y$ axis at $\EE{V^{\pi^*}(s^{b})}$.
Since, obviously, $s^{b} \succeq s^{b'}$ for any $b > b'$, w have $V^{\pi^*}(s^{b}) \le V^{\pi^*}(s^{b'})$, and so $\EE{V^{\pi^*}(s^{b})}$ is non-increasing in $b$. Therefore, there exist thresholds $0=\thresh_{B+1} \leq \thresh_{B} \leq \cdots \leq \thresh_1 \leq C_{max}$, such that the optimal action is $A^{b}$ if the channel cost belongs to the interval $[\thresh_{b+1},\thresh_{b}]$.
Since $A^{b'} \subset A^{b}$ for $b > b'$, this also means that the simple action $a^b=(l_{b},L_{b})$ is performed whenever $C \le \thresh_b$ (note that $\thresh_b=0$ means that action $a^b$ is never performed because $C>0$).

This implies the following theorem.

\begin{theorem}
\label{th:opt_char}
Consider a state $s=(\rel,\cache,E) \in \S$ of the MDP-SI for the proactive caching problem, and let $C$ denote the channel cost. Let $l_1\le \cdots \le l_B$ denote the contents in $\cache$, and $L_1 \ge \cdots \ge L_B$ denote the $B$ largest elements of $\rel$. Then, there exist threshold values $0 \leq \thresh_{B} \leq \thresh_{B-1}\leq \cdots \leq \thresh_1 \leq C_{max}$, such that there is an optimal caching policy that performs the simple actions $ a^i = (l_i|L_i)$ for all $i$ such that $C\le \thresh_i$ if $E>0$ (i.e., the user does not access the OSN).
\end{theorem}


The thresholds for different simple actions depend on what other simple actions are available, and on the contents of the cache. This is because if we cache a content, its value depends on the likelihood of the content to be removed from the cache before being consumed by the user, and this likelihood is affected by the lifetime of the other contents in the cache.

Having shown that the optimal policy exhibits a threshold behavior, one still has to evaluate the optimal threshold values in order to characterize the optimal policy and the corresponding optimal performance. Therefore, the number of threshold values to be determined is in the order of the cardinality of the state space $\S$, which is extremely large. This makes it computationally infeasible to compute the optimal threshold values.  Interestingly, this is not the case if we have a sufficiently large cache capacity, e.g., $B \ge M_{\max} K_{\max}$, in which case we never remove a content from the cache unless it is consumed, or has expired; we will refer to this as the case of unlimited cache capacity. Therefore, we can decide about each content individually, and independently of the cache contents. This gives rise to the following corollary:
\begin{corollary}
Assume that the cache capacity is unlimited, that is, $B=\infty$. Then, for any state $s=(\rel,\cache,E)$ with $E>0$, there exist thresholds $0\le \thresh_{1,E} \le \cdots \le \thresh_{K_{\max},E} \le C_{max}$ (recall that $K_{\max}$ is the maximum lifetime), which depend only on $E$, such that a content with remaining lifetime $L \in \rel$ is downloaded if $C \le \thresh_{L,E}$.
\label{cor:LB-UC}
\end{corollary}

Since the decision to download any content is independent of the others, the problem can be modeled as a finite-horizon MDP-SI, where the horizon equals the remaining lifetime $L$ with maximum horizon $K_{\max}$. Thus, we can apply dynamic programming \cite{Bertsekas:book} to determine the optimal downloading thresholds recursively: Let $V_{L,E}$ denote the future download cost associated with a content with lifetime $L$ from a state with time $E$ elapsed since the past user access following an optimal policy. Since there is no need to proactively download a content with lifetime $1$, $\thresh_{1,E}=0$ for all $E>0$, and so $V_{1,E}=0$ for any $E>0$. Let $p_E=\PPcond{U_{t+1}=1}{E_t=E}
=\PP{D_1 \le E+1} - \PP{D_1 \le E}$ denote the probability of user access in the next time slot.\footnote{Recall that $D_n$ denotes the $n$th inter-access time and the $D_n$ are i.i.d.}
Assuming that optimal decisions will be made for lifetimes up to $L-1$ for all $E$, a decision with threshold $\thresh$ for lifetime $L>1$ and elapsed time $E$ has a future download cost 
\begin{align}
\label{eq:vle}
V_{L,E,\thresh} &= \PP{C \le \thresh} \EEcond{C}{C\le\thresh}
+ \PP{C>\thresh}\big(p_E \E[C] + (1-p_E)V_{L-1,E+1} \big)~.
\end{align}
Minimizing the above expression in $\thresh$ by setting its derivative to zero, we obtain that the optimal threshold is $\thresh_{L,E}=p_E \EE{C} + (1-p_E) V_{L-1,E+1}$, which is exactly the expected future download cost if the content is not downloaded in the current state. Noticing that $\thresh_{L,E}$ equals the last term in parentheses in \eqref{eq:vle}, we obtain the following result.
\begin{corollary}
\label{cor:LB-UC2}
Assume that the cache capacity is unlimited (i.e., $B_{\max}=\infty$). Then the optimal thresholds $\thresh_{L,E}$ can be computed recursively as follows: $\thresh_{1,E}=0$ for all $E>0$. For $L \ge 1$, given $\thresh_{L,E}$ for all $E$, the optimal thresholds for $L+1$ can be obtained for all $E$ as 
\[
\thresh_{L+1,E}=p_E \EE{C} + (1-p_E)\bigg(
\PP{C \le \thresh_{L,E+1}} \EEcond{C}{C\le\thresh_{L,E+1}}
+ \PP{C>\thresh_{L,E+1}} \thresh_{L,E+1}\bigg)~.
\]
\end{corollary}


In case of the IRM user access model, the same thresholds can be used in all states, and the expression for the thresholds simplifies to $\thresh_1=0$ and for $L \ge 1$,
\begin{equation}
\label{eq:LB-UC3}
\thresh_{L+1}=p_a \EE{C} + (1-p_a)\bigg(
\PP{C \le \thresh_{L}} \EEcond{C}{C\le\thresh_{L}}
+ \PP{C>\thresh_{L}} \thresh_{L}\bigg)~.
\end{equation}

The optimal performance with an infinite cache capacity will be studied as a lower bound on the optimal performance for a practical finite cache capacity system in Section \ref{sec:lower_bound}.


\section{Low-Complexity Caching Schemes via Policy Approximation}
\label{sec:approx_schemes}

In the previous section we determined the structure of the optimal caching policy. 
According to Theorem~\ref{th:opt_char}, the optimal policy has a threshold structure, and the threshold for each simple action depends in general on the remaining lifetimes of all the contents inside and outside the cache, as well as on the time elapsed since the last user access. This implies that the optimal policy may employ completely different threshold values for the same simple action at different system states, and the optimal policy belongs to the family of policies parametrized by these thresholds. The dimension of this policy set is $|\bar{\S}|$, where $\bar{\S} \subset\S$ denotes the set of system states where the user does not access the OSN. Moreover, we have approximately $K_{\max}^2/2$ potential simple actions, each of which can have a different threshold value at each state. However, for any reasonable (finite) cache size $B$, this is huge; and hence, it is infeasible to compute an optimal policy (e.g., if $M_{\max} \ge B$, then just the cache content $\cache$ can take $\binom{B+K_{\max}}{K_{\max}}$ different values, which is already prohibitively large for even moderate values of $B$ or $K_{\max}$).
To resolve this problem, we use policy approximation techniques, and approximate the policy space using some simple parametrized form.



From now on we adopt the IRM user access model, which alleviates the need to consider the time $E_t$ elapsed since the last user access, readily reducing the state space. In the rest of this section, we introduce two low-dimensional approximations to the policy space, which allow us to run optimization algorithms (policy search algorithms, described in Section~\ref{sec:pgm})  to find a policy with good performance, and hence, give rise to computationally feasible caching schemes. These schemes are not based on a priori known statistics of the system and optimize the policy parameters based on observations (these observations can be collected either from the real system or via simulations through a generative model). Therefore, in principle, the methods can be used in a learning context, where an agent, who does not known the statistics of the environment a priori, can learn from its actions to update its policy in order to adapt to the unknown environment in a reinforcement learning fashion.

\subsection{The longest lifetime in--shortest lifetime out (LISO) policy}
\label{subsec:liso}
The longest lifetime in--shortest lifetime out (LISO) policy is a suboptimal threshold-based proactive caching policy with a simplified structure, such that it has a single threshold value for each simple action (corresponding to the content pair consisting of the content with the shortest remaining lifetime inside the cache and the one with the longest remaining lifetime outside the cache), independent of the system state. For every such pair, if the channel cost is below the threshold value, the two contents are ``swapped,'' and no action is taken otherwise. In this case, the policy is directly parametrized by the threshold values. That is,
\begin{equation*}
\thresh(l|L) = \theta(l,L),
\end{equation*}
where $\theta(l,L) \in [0,C_{\max}]$, for all pairs $a = (l|L)$, $l,L \in \{0, \ldots, K_{\max}\}$. Thus, the set of policies parametrized this way is of dimension $(K_{\max} + 1)^2$, which is feasible. We can further reduce the dimension by explicitly forbidding all simple actions $(l|L)$ with $l \ge L$ (i.e., setting the corresponding $\theta(l,L)$ to zero), since an optimal policy will not replace a cached content with a content that has a shorter remaining lifetime. Hence, for such simple actions, we have $\thresh(l|L) = 0$. We also note that the optimal policy has a monotonic structure; that is, $\thresh(l|L_1)\le \thresh(l|L_2)$ if $L_1 < L_2$ and $\thresh(l_1|L)\ge \thresh(l_2|L)$ if $l_1 < l_2$. These observations can be used to limit the search space for the threshold vector, speeding up the policy search methods.

\subsection{Linear function approximation (LFA) policy}
\label{subsec:lfa}
Next, we propose an improved policy representation (an extension of LISO), which takes into account the remaining lifetimes of the contents in the cache memory when determining the threshold values; this information can be useful in estimating the likelihood that a downloaded content will be removed from the cache before it expires or is consumed. However, to keep the computational complexity feasible, we employ linear function approximation (LFA) \cite{Sutton:1998:book}. To characterize the state of the cache, we define \emph{features} of the cache-state based on the number of contents in the cache with a particular remaining lifetime (this is meaningful thanks to the homogeneity of the size of the contents). More precisely, we define the cache-state features by a \emph{frequency vector} $\Phi_t = \left[\phi_t(0), \phi_t(1), \ldots, \phi_t(K_{\max})\right]$, where $\phi_t(i)$ is the ratio of the number of contents with lifetime $i$ in the cache at time $t$, that is,
\begin{equation}
\phi_t(i) \triangleq \frac{\sum_{l \in \cache} \ind{l = i}}{B}, \quad\mbox{for}\quad i = 0,1,\cdots,K_{\max},
\label{freq_component}
\end{equation}
where $l = 0$ denotes the empty locations as before. Clearly, $0 \le \phi(i) \le 1$, and $\sum_{i=0}^{K_{\max}}\phi(i) = 1$.

The threshold value for each simple action $a(l|L)$ for $l<L$, $l,L \in \{0, \ldots, K_{\max}\}$, is then defined as a linear function of the frequency vector $\Phi$ as
\begin{equation}
\label{LFA}
\thresh(l|L) = \sum_{i=0}^{K_{\max}}\phi(i) \theta_i(l,L) = \Phi^\top \mathbf{\theta}(l,L),
\end{equation}
where $\theta_i(l,L) \in \R$ for $l < L$, and $\theta_i(l,L)=0$ otherwise. The resulting scheme defines a $K_{\max}(K_{\max}+1)^2/2$-dimensional policy space.

\begin{remark}
We note that the LISO policy described in Section~\ref{subsec:liso}, which is directly parametrized by the threshold values for each simple action ignoring other contents in the cache, is a special case of the LFA policy with parameters $\theta_i(l,L) = \theta(l,L)$ for all $i$.
\end{remark}
In the next section, we describe two policy search algorithms that we use to optimize the parameters of the proposed approximate caching schemes.

\section{Policy Search Methods}
\label{sec:pgm}
Optimizing parametric policies for MDPs has been extensively studied in reinforcement learning \cite{Sutton:1998:book}. In this paper, we are going to employ policy gradient (PG) methods to optimize the parameters of our LISO and LFA policies. This section, based on \cite{Policysearch}, introduces these algorithms. PG methods are model-free reinforcement learning algorithms to find an optimal policy in an MDP by running gradient descent over the policy space to minimize the expected average cost $\rho^{\pi_{\para}}$, where $\pi_{\para}$ denotes the policy defined by the parameter vector $\para$. 
That is, in every step of the policy gradient algorithm, the actual parameter $\para_j$ is updated using the gradient $ \nabla_{\para} \rho^{\pi_{\para}}$ of $\rho^{\pi_{\para}}$ according to $\para$ as
\begin{align}
\para_{j+1} = \para_j - \lambda \nabla_{\para} \rho^{\pi_{\para_j}},
\label{eq:polupdate}
\end{align}
for some positive step size $\lambda$.

Since the gradient $\nabla_{\para} \rho^{\pi_{\para_j}}$ is not known in closed form in most cases, the gradient (and the average cost of the policy) has to be estimated  through sample averages over independent, finite trajectories of the system obtained via Monte Carlo rollouts. This, in turn, implies that, instead of \eqref{eq:polupdate} we will use a random estimate of the gradient; so, in practice, \eqref{eq:polupdate} will become a stochastic gradient descent algorithm. To curtail the effect of noise that is introduced due to the randomness, we obtain $\para_{j+1}$ as the average of $m$ policy updates, i.e., $\para_{j+1} = \frac{1}{m}\sum_{i=1}^m\para_{j+1,i}$, where each $\para_{j+1,i}$ is obtained using \eqref{eq:polupdate} with an independent estimate of the gradient.
The estimation procedure usually requires two steps:
\begin{enumerate}
\item \emph{Policy evaluation}: The average cost of a sample trajectory $\tau_{\para} =(S_1,C_1,A_1), \ldots, (S_T, C_T, A_T)$, obtained by following a given policy $\pi_\para$ with parameter vector $\para$, is evaluated as
\begin{equation}
J(\tau_\para) = \frac{1}{T}\sum_{t=1}^T \mu(S_t,A_t,C_t).
\label{eq:av_samplecost}
\end{equation} 
\item \emph{Policy exploration}: New sample trajectories are generated. Exploration is implemented either directly on the actions $A_t$, or on the parameter vector $\para$, by introducing an exploration \emph{noise} either at every time step of the trajectory, or at the beginning of the trajectory. 
\end{enumerate}

In what follows, we review two practical policy gradient algorithms that employ different estimation techniques.

\subsection{Finite difference method (FDM)}
In FDM, we estimate the gradient by generating sample trajectories following the policy $\pi_{\para_j}$ given by the parameter vector $\para_j$ (determining the threshold values $\thresh(l|L)$) and by new policies obtained by applying small perturbations $\Delta{\para}^{[i]}$ to the parameter vector $\para_j$. Generating trajectory $\tau^{[i]}$ for $\para_j$ and $\tau^{[i]}_\Delta$ for $\para_j + \Delta\para^{[i]}$, the change in the cost is estimated by
\begin{equation*}
\Delta J^{[i]} = J(\tau^{[i]}_\Delta) - J(\tau^{[i]}),
\end{equation*}
which is also approximately equal to $(\nabla_{\para} {\rho^{\pi_{\para}}})^\top \Delta\para^{[i]}$.
Thus, generating $N$ independent trajectories $\tau^{[i]}$, for $i=1,\ldots,N$, the gradient can be estimated from $\Delta\boldsymbol J_{\pi_\para}= [\Delta J^{[1]},\cdots, \Delta J^{[N]}]^\top$ and $\Delta\boldsymbol\Theta = [\Delta{\para}^{[1]},\cdots, \Delta{\para}^{[N]}]^\top$ by linear regression as follows:
\begin{align}
\nabla_{\para} \rho^{\pi_{\para}} \approx \left(\Delta\boldsymbol\Theta^\top \Delta\boldsymbol\Theta\right)^{-1} \Delta\boldsymbol\Theta^\top \Delta\boldsymbol J_{\pi_\para}~.
\end{align}

In the FDM method, policy exploration is implemented on the parameter vector at the beginning of each trajectory. Perturbations can be chosen randomly; in this paper the perturbations for each coordinate of $\para_j$ are drawn from a uniform distribution within the range $[-r,r]$, for some relatively small positive real number $r$.



\subsection{Likelihood-ratio method (LRM)}
In this section we describe another PG strategy, called  LRM. To simplify the treatment, throughout we use the MDP notation and include the cost $C$ in the state $S$. 

In LRM, exploration is implemented directly on the actions, and in every time step of each trajectory by using a randomized policy $\pi_{\para}(A\vert S) \in[0,1]$, which takes action $A$ in state $S$ with probability $\pi(A|S)$. Since $A$ may consist of several simple actions, for each simple action $(l|L)$ where $l<L$, $l,L \in \{0, 1, \ldots, K_{\max}\}$, we define a randomized policy $\pi_{\para}((l|L)\vert S)$ as a sigmoid function with negative slope parameter:
\begin{align*}
\pi_{\para}((l|L)\vert S) = \frac{1}{1+e^{-\eta(\thresh(l\vert L)-C)}}~,
\end{align*}
where $\eta > 0$ is the slope factor. Given cache contents $l_1\le \cdots \le l_B$, and the $B$ contents outside the cache with the largest remaining lifetimes $L_1 \ge \ldots \ge L_B$, we repeatedly try to perform the action $a^i \triangleq (l_i|L_i)$ with probability $\pi_{\para}(a^i|S)$ for $i=1,\ldots, B$, until the first failure. This implies that for $B' \le B$, the probability of performing action $A_{B'}=\{a^1,\ldots,a^{B'}\}$ is 
\begin{align*}
\pi_{\para}(A^{B'} \vert S) = \big(1-\pi_\para(a^{B'+1}|S)\big) \prod_{i=1}^{B'}\pi_{\para}(a^i\vert S),
\end{align*}
where $\pi_\para(a^{B+1}|S)$ is defined to be zero for all states $S$.

Let $P_{\para}$ denote the density of an infinite trajectory $\tau=(S_1,A_1),(S_2,A_2),\ldots$ obtained by following policy $\pi_\para$, and let $J(\tau)=\limsup_{T \to \infty}
\frac{1}{T} \sum_{t=1}^T \mu(S_t,A_t)$. Then, under general, non-restrictive assumptions, we have
\begin{align*}
\nabla_{\para} \rho^{\pi_{\para}} = \int \nabla_{\para} P_{\para}(\tau)J_{\pi_\para}(\tau)d\tau~.
\end{align*}
Using the ``likelihood-ratio'' identity $\nabla_{\para}\log P_{\para}(\tau) = \nabla_{\para} P_{\para}(\tau)/P_{\para}(\tau)$, the above gradient can be expressed as
\begin{align}
\label{eq:LRMgrad}
\nabla_{\para} \rho^{\pi_{\para}} &=\int P_\para(\tau)\nabla_{\para}\log P_\para(\tau)J(\tau)d\tau
= \E[\nabla_{\para}\log P_\para(\tau)J(\tau)]~.
\end{align}

The expectation with respect to the trajectory distribution $P_{\para}$ is approximated by sample averages over sampled trajectories $\tau^{[i]}$ of finite length. Interestingly, as it is well known in the reinforcement learning literature \cite{Policysearch}, this can be done without the knowledge of the density $P_\para$. Indeed, since $P_\para(\tau) = P(S_1)\prod_{t=1}^TP(S_{t+1}\vert S_t,A_t)\pi_{\para}(A_t\vert S_t)$, taking logarithm and differentiating according to $\para$ gives
\begin{align}
\nabla_{\para}\log P_{\para}(\tau) = \sum_{t=1}^T \nabla_{\para}\log \pi_{\para}(A_t\vert S_t),
\label{eq:log_distribution}
\end{align}
which can be computed directly from $\tau$ using the parametric form of $\pi_\para$. Thus, the expectation in \eqref{eq:LRMgrad} can be estimated by averaging over a number of independent trajectories sampled following policy $\pi_\para$.
To minimize the variance of the estimate, we introduce a \emph{baseline} vector $\mathbf{b}$, as in the REINFORCE algorithm \cite{REINFORCE90}, and estimate the $h$th coordinate of the gradient by
\begin{align*}
\nabla_{\para_h}\rho^{\pi_{\para}} = \E \left[\sum_{t=1}^T \nabla_{\para_h}\log \pi_{\para}(A_t\vert S_t)(J(\tau) - b_h)\right].
\end{align*}
The baseline does not introduce any bias in the gradient estimate: using \eqref{eq:log_distribution}, the likelihood-ratio identity, and the fact that $\int \nabla_{\para_h}P_{\para}(\tau)d\tau=0$ since $\int P_{\para}(\tau)d\tau = 1$, we get
\begin{align*}
\E \left[\sum_{t=1}^T \nabla_{\para_h}\log \pi_{\para}(A_t\vert S_t)b_h\right] = b_h\int \nabla_{\para_h}P_{\para}(\tau)d\tau = b_h\nabla_{\para_h}\int P_{\para}(\tau)d\tau = 0.
\end{align*}
As in the REINFORCE algorithm, we select the baseline $b_h$ for $\nabla_{\para_h}\rho^{\pi_{\para}}$ by minimizing the variance of the estimate of the $h$th coordinate, which yields
\begin{align*}
b_h = \frac{\E \left[\left(\sum_{t=1}^T \nabla_{\theta_h}\log \pi_{\para}(A_t\vert S_t,C_t)\right)^2 J(\tau)\right]}{\E \left[\left(\sum_{t=1}^T \nabla_{\theta_h}\log \pi_{\para}(A_t\vert S_t,C_t)\right)^2\right]}~.
\end{align*}

Numerical results obtained with the proposed low-complexity proactive caching algorithms optimized with both of the PG methods will be presented in Section \ref{sec:numerical}. Next, we present lower bounds on the performance to evaluate the performance loss introduced by the proposed low-complexity caching policies.



\section{Lower Bounds}\label{sec:lower_bound}

We derive two lower bounds on the average cost, which become tight under different settings of the problem.  
The lower bounds are based on some relaxation of the system constraints (the first one assume infinite cache capacity while the second one a non-causal knowledge of the user access times). For both cases, in the next subsections, we derive the optimal caching policy, and the lower bounds can be computed by estimating the performance of the optimal policy (in the corresponding relaxed system) using Monte Carlo rollouts; that is, averaging $J(\tau_\para)$ in \eqref{eq:av_samplecost} for several independent trajectories $\tau_\para$ obtained by following the optimal policy $\pi_\para$.

\subsection{Lower bound with unlimited cache capacity (LB-UC)}
\label{subsec:LBUC}
The first lower bound, called LB-UC, is obtained by considering an unlimited cache capacity (i.e., $B=\infty$). In this case, as explained in detail in Section \ref{sub:OptimalPolicy}, there is no need to remove and replace any content inside the cache. The decision to download and store a content can be taken individually, and independently of the existing contents in the cache. 
The structure of the optimal threshold values for LB-UC follows from Corollary \ref{cor:LB-UC2} and \eqref{eq:LB-UC3}. 

\subsection{Lower bound with non-causal knowledge of the user access times (LB-NCK)}
\label{subsec:LBNCK}
The second lower bound, called LB-NCK, is obtained by assuming non-causal knowledge of the user access times. Since the user access times are known in advance, contents that will expire before the user accesses the OSN will never be downloaded and can be automatically removed from the system. Therefore, there is no need to remove any content from the cache, thus $A_t^{(2)} = \emptyset$ for all $t$. All the remaining contents must be downloaded by the next user access, which means that there is no difference among contents, and so whenever the CM decides to download a content, it does not matter which one it is. As such, without loss of generality, we assume that the CM may pre-cache only the first $B$ contents that become available after a user access. This means that there will always be space in the cache for these contents to be downloaded. Then it follows that the optimal policy only needs to decide when to download these contents, and this decision is independent of the contents in $\rel$ and $\cache$, and may only depend on the time till the next user access and the current channel cost.

To determine the optimal policy, similarly to Corollary~\ref{cor:LB-UC2} and \eqref{eq:LB-UC3}, the problem can be modeled as a finite-horizon MDP-SI, where the time horizon is the time $G$ until the next user access. Denoting by $V^{NCK}_G$ the average energy cost of downloading a content following the optimal policy; we have $V^{NCK}_0=\EE{C}$, since the content must be downloaded when $G=0$. For any $G \ge 1$, the dynamic programming equations imply that $V^{NCK}_G=\EE{\min\{C,V^{NCK}_{G-1}\}}$. Therefore, the optimal decision is to download a content if the channel cost $C$ is smaller than the future download cost $V^{NCK}_{G-1}$. Thus, the optimal policy again has a threshold structure.
\begin{corollary}
Assuming that the user access times are known non-causally, there exist thresholds $C_{max} \ge \thresh^{NCK}_{1} \ge \cdots \ge \thresh^{NCK}_{D_{max}} \ge 0$\footnote{Note that $D_{max}$ is the bound on the length of the user access interval, and can be infinite under the IRM model.} such that, 
for any $\rel,\cache$, if there are $G$ time slots left until the next user access,  
a content with remaining lifetime $L\in\rel$ (with $L \ge G$) is downloaded for $G \ge 0$ (i.e., when $U_t=0$) if $|\cache| < B$ and $C_t \leq \thresh_{G}$.
The thresholds are given by the recursion $\thresh^{NCK}_1=\EE{C}$ and for $G \ge 2$,
$\thresh^{NCK}_G=\EE{\min\{C,\thresh^{NCK}_{G-1}\}}$.
\label{cor:LB-NCK}
\end{corollary}

\section{Numerical Results} 
\label{sec:numerical}
Here we present numerical simulations implementing the proposed caching schemes with both FDM and LRM. We compare their performances with the two lower bounds in Section \ref{sec:lower_bound}, as well as with reactive and random caching schemes. In \textit{reactive caching}, all the relevant contents are downloaded at the time of user access. This scheme does not utilize the storage available at the UE. In \textit{random caching}, when $U_t=0$, each relevant content in $\rel$ is downloaded randomly, with a constant probability $p_r > 0$ whenever $|\cache_t| < B$, and $p_r = 0$ whenever $|\cache_t| = B$. This scheme does not utilize any intelligence in making caching decisions. 
Note that random caching is equivalent to reactive caching when $p_r = 0$. While random caching exploits the cache memory, we will see that this does not improve the performance since it is done in a random manner.


\subsection{System Setup}
\label{subsec:sim_setup}
The number of contents generated at each time slot, $M_t$, is drawn uniformly at random from the set $\{1,\ldots,M_{max}\}$, while the lifetime $K_{t,i}$ of individual contents $i \in \{1,\ldots,M_t\}$ at the time of generation is drawn from the set $\{5,10,\ldots,K_{\max}\}$, where $K_{\max}$ is a multiple of $5$. We assume that the user accesses the system independently at each time slot, with probability $p_a=0.25$.

We obtain $C_t$ using Shannon's capacity formula, $R = W \log_2\left(1 + P_{signal}/P_{noise}\right)$, where $R$ is a deterministic transmission rate, $W$ is the channel bandwidth, $P_{noise}$ is the noise power and $P_{signal}$ is the signal power. Using system parameters consistent with the Long Term Evolution (LTE) network model \cite{Sesia:LTE:2011}, on a dB scale, the noise power is given by
\begin{equation}
P_{noise} = 10\log_{10}(kT) + 10\log_{10}W + NF,
\label{eq:noise_pow}
\end{equation}
where $kT = -174$ dBm/Hz is the noise power spectral density, and $NF = 5$ dB is a typical noise figure. The signal power is given by
\begin{equation}
P_{signal} =  C_t + G_{TX} + G_{RX} - PL(d),
\label{eq:signal_pow}
\end{equation}
where $G_{TX}$ and $G_{RX}$ are the transmit and receive antenna gains, respectively, and $PL(d)$ is the pathloss, which is a function of the distance $d$ between the user and the serving BS. We adopt the 3GPP channel model \cite{3GPP_release9}, and consider an urban micro (UMi) system, with an hexagonal cell layout in the non-line-of-sight (NLOS) scenario, in which case the pathloss is given as 
\begin{equation}
PL(d) = 36.7\log_{10}(d) + 22.7 + 26\log{10}(f_c) + \mathcal{X}_\sigma,
\label{eq:pathloss}
\end{equation}
where $f_c = 2.5GHz$ is the center frequency, and $\mathcal{X}_\sigma$ is the shadow fading parameter drawn from a zero-mean log-normal distribution with standard deviation $\sigma = 4$ dB. The distance $d$  is in meters ($\mathrm{m}$), and we assume that the user location, described by $d$, is uniformly distributed in time and across BSs.

We assume that a micro BS has a radius of $250\mathrm{m}$, and the shortest possible distance of a user from a serving BS is $50\mathrm{m}$. Therefore, the user distance $d$ from the serving BS in any time slot is drawn from a uniform distribution $d \sim \mathcal{U}(50,250)$. We assume that the user is only served by a single BS in every time slot. Although we focus on a single user, the savings in energy will scale proportionally with the number of users. We use the values $G_{TX} = 17$ dBi and $G_{RX} = 0$ dBi. To compute the noise power, we assume a fixed (average) bandwidth of $10$ MHz in every time slot, and for the Shannon capacity formula, we assume a spectral efficiency of $R/ W = 2~\mathrm{bps/Hz}$ for each content item. The required power will be linearly scaled with the number of contents downloaded at each time slot, assuming they are independently encoded and transmitted over orthogonal subbands. For all the simulations, we set the initial state as $\rel_0 = \cache_0 = \emptyset$ and $E_0 = 0$. The cache capacity $B$ is measured in number of contents.


For the FDM algorithm, we select the perturbation parameters $\Delta\theta$ from a uniform distribution $\Delta \theta \sim \mathcal{U}(-0.08, 0.08)$. For each iteration, a policy update is performed after $100$ trajectories, with the duration of a trajectory set as $300$ time slots. On the other hand, for the LRM algorithm, for the randomized policy to closely resemble the actual deterministic policy, the logistic function defining the policy should be as close to a unit step function as possible. Hence, we choose the slope $\eta = 10$. A policy update is performed after only $20$ trajectories, with the duration of a trajectory set as $300$ time slots.

For the initial parameter vector $\para_0$ of the LISO policy, we use the threshold values obtained as the solution of the unlimited cache capacity problem as the initial $\theta(0,L)$ values for all $L$. For LFA, we use the components of the same parameter vector (obtained from the unlimited cache capacity problem) as the initial components $\theta_i(l,L), \forall l,L, ~\mbox{and}~ i\in\{0,1,\ldots,K_{\max}\}$. This initialization allows our algorithm to start from a relatively good initial point, thus improving the convergence speed over a random initial parameter vector. For all the algorithms, an average of $5$ policy updates is taken as the policy update of any iteration. In each simulation setup, we select an appropriate step size by adjusting the step size at different runs until the best result is obtained. Finally, to test the performance of any algorithm and policy, we use a \emph{test data} of $100$ trajectories, each consisting of $5000$ time slots.

\subsection{Performance Evaluation}
\label{subsec:perform_eval}

\begin{figure}[!tbp]
\centering
  \subfloat[]{\includegraphics[scale=0.43]{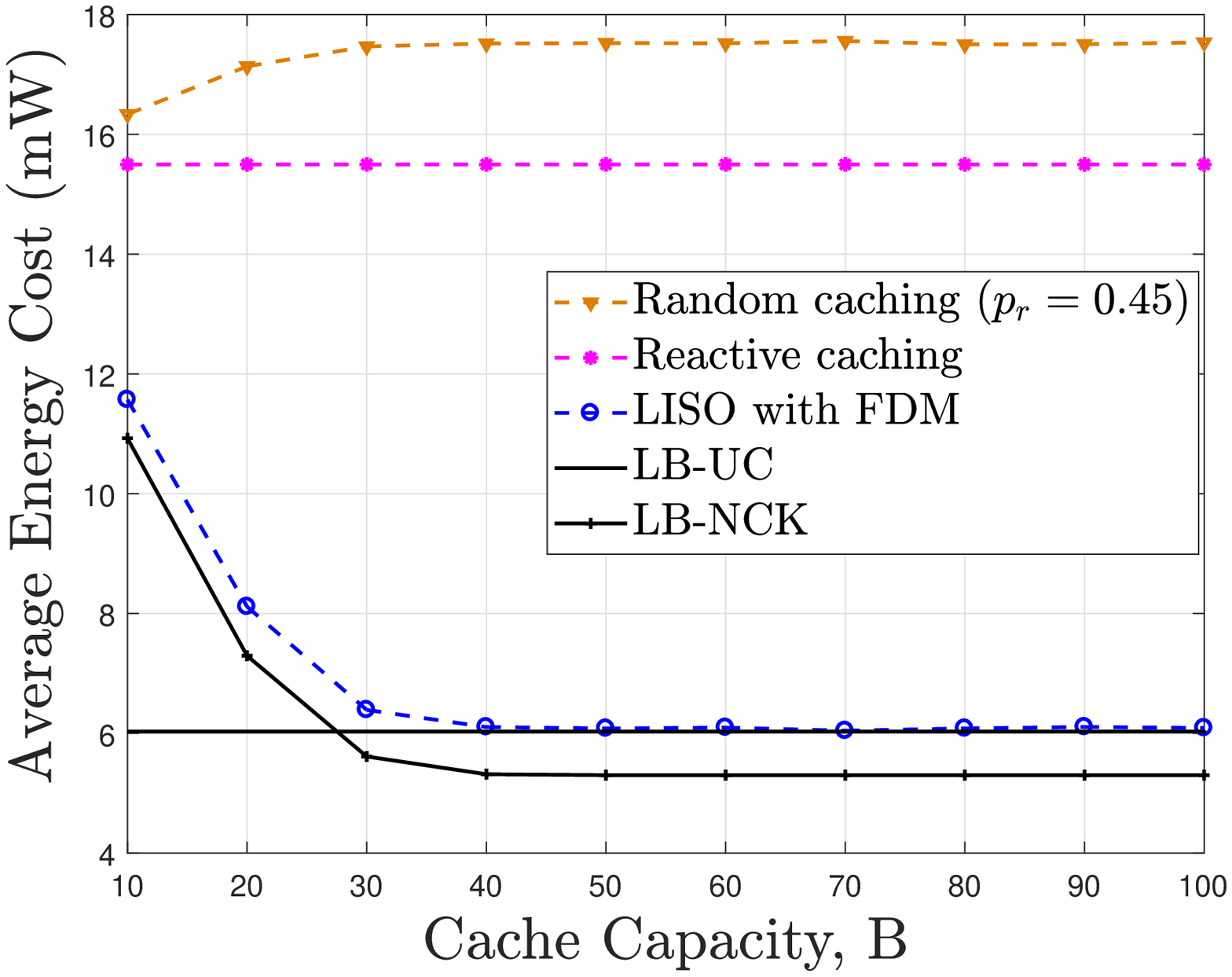}\label{fig:costvscap_LISO}}
  \subfloat[]{\includegraphics[scale=0.43]{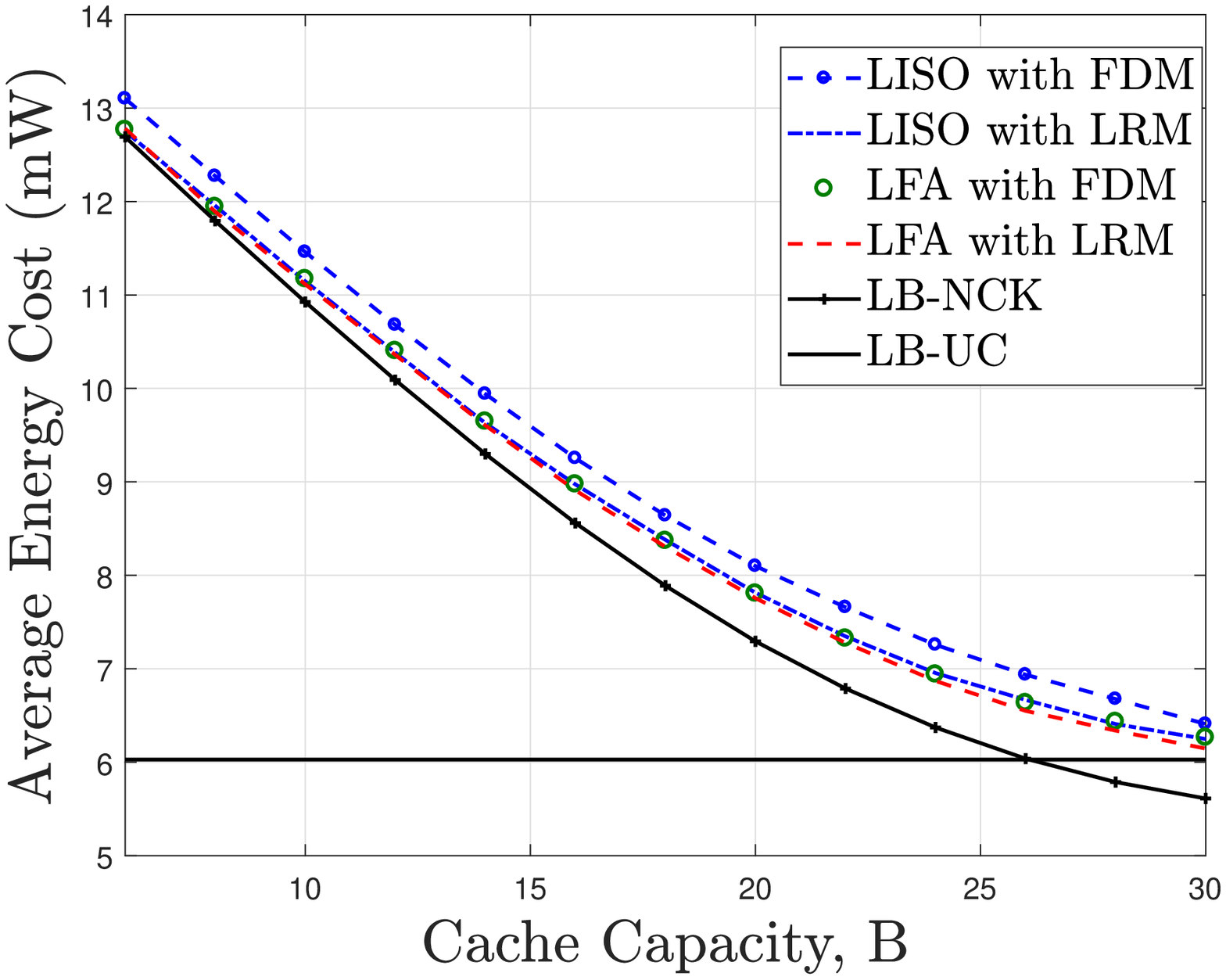}\label{fig:costvscap_All}}
  \caption{Average energy cost vs. cache capacity with $K_{\max} = 15, M_{max} = 8, p_a = 0.25$.}
\vspace{-0.5cm}
\end{figure}

To simplify the presentation, we first compare the LISO policy implemented with the FDM algorithm, with the benchmarks in Fig. \ref{fig:costvscap_LISO}, where we plot the average energy cost with respect to the cache capacity. 
We use $p_r = 0.45$ for the random caching scheme. We observe that the random caching scheme has the highest average energy cost, which increases with $p_r$. The average energy cost of the reactive scheme is independent of the cache capacity as it does not utilize the caches. While the reactive scheme only downloads contents that are actually requested, the random scheme downloads many contents that will eventually expire before being requested by the user. Moreover, since the random scheme downloads the contents randomly, it has, on average, the same cost as the reactive scheme for the contents that are eventually consumed by the user. As expected, the performance of LB-UC, which assumes unlimited cache capacity, does not depend on the cache capacity; while that of LB-NCK, which assumes non-causal knowledge of the user access times, decreases with the cache capacity. This is because more contents that will remain relevant by the user access time can be downloaded at favorable channel conditions through proactive caching. The performance of LB-NCK is equal to that of the reactive scheme when there is no cache memory available at the user, i.e., $B = 0$, since proactive caching is not possible in that case. 

The proposed LISO scheme significantly improves the system's performance with respect to the reactive caching scheme for any nonzero cache capacity. For a cache capacity of $B=30$, the LISO policy achieves an approximately $60\%$ reduction in energy consumption over the reactive scheme. For relatively large cache capacities, i.e., $B \geq 40$, the performance of the LISO scheme almost meets that of LB-UC. This is because more contents that will not expire by the time the user accesses the system can be stored in the cache, and almost no contents need to be removed from the cache.  This means that a cache capacity of $B=40$ is sufficient to provide all the potential gains from proactive caching in this setup. Interestingly, $40$ is roughly the average number of relevant contents at any point in time. Moreover, in the low cache capacity regime, the performance of the LISO scheme is very close to that of LB-NCK. This is because when the cache capacity is small, the system is relatively conservative in proactively caching contents, and so downloaded contents seldom expire or are swapped out of the cache before the next user access. Thus, the gain from knowing the user access times is more limited.
We conclude from Fig. \ref{fig:costvscap_LISO} that any  improvement in the performance of the LISO policy, when implemented with the FDM algorithm, can only occur at the low cache capacity regime ($B \leq 30$).

In Fig. \ref{fig:costvscap_All} we plot the performance of both caching schemes, LISO and LFA, implemented with both FDM and LRM algorithms, in the low cache capacity regime. 
We observe that both the policy representation using LFA and using the LRM algorithm for gradient estimation improve the performance compared to LISO with FDM. At very low cache capacities, i.e., $B < 10$, the performances of LISO with FDM, and LFA with FDM or LRM all closely follow the LB-NCK bound. Meanwhile, the LFA policy has a performance gain of up to $4.4\%$ over the LISO policy when both schemes are implemented with the FDM algorithm. This performance gain can be attributed to the fact that the LFA policy considers the remaining lifetimes of all the contents inside the cache when making a cache decision, which is ignored by the LISO policy. When the LFA policy is implemented with LRM, it achieves a performance gain of up to $5.6\%$ over the LISO policy implemented with FDM. LRM also improves the performance of LISO policy with up to $4.2\%$ with respect to LISO with FDM. We can attribute the better performance of LRM to its improved exploration strategy. 

\begin{figure}[t]
\centering
\includegraphics[scale=0.45]{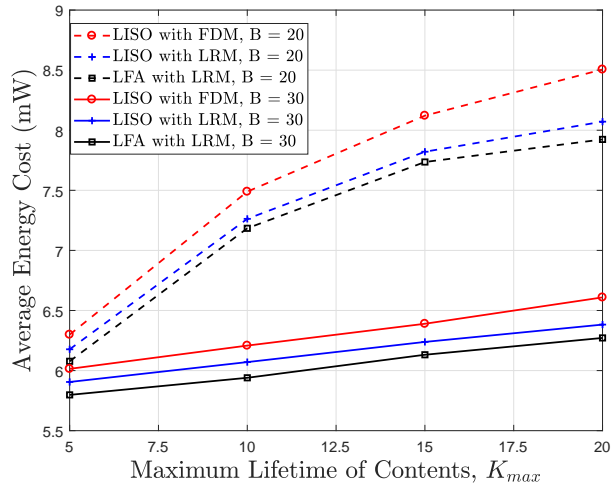}
\caption{Average energy cost vs. maximum lifetime of contents for $B=20, 30$, when $M_{max} = 8, p_a = 0.25$.}
\label{fig:costvslifetime}
\vspace{-0.5cm}
\end{figure}

In Fig. \ref{fig:costvslifetime} we plot the average energy cost against the maximum lifetime of contents, $K_{\max}$. We observe that the energy cost increases with the lifetime of contents. This is expected: when the contents remain relevant longer, more contents will be consumed by the user at the time of access. We observe that the improved performance of LFA policy over LISO extends to the $K_{\max}$ values considered here. 
The performance gain of LISO with LRM with respect to LISO with FDM increases with $K_{\max}$, which means that a better exploration strategy becomes more important as $K_{\max}$ increases, since the cache space becomes relatively more limited per relevant content.

In Fig. \ref{fig:convergence_rate} we compare the convergence rates of the two PG methods. We observe that, initially, LRM performs worse than FDM. However, after about $250$ trajectories, LRM starts to converge at a faster rate, saturating to the optimal performance after approximately $1000$ trajectories. The LRM is known to have better theoretical convergence guarantees, and its superiority over FDM has been observed in other applications as well \cite{Peters:2006}.  


\begin{figure}[t]
\centering
\includegraphics[scale=0.45]{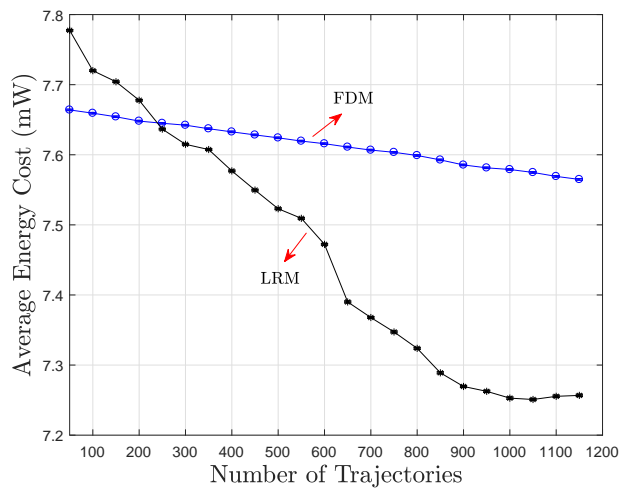}
\caption{The evolution of the FDM and the LRM algorithms with respect to the number of rollouts.}
\label{fig:convergence_rate}
\end{figure}

In the next section, we go beyond our modeling assumptions to show that the performance gain of the LFA policy with respect to LISO can be more significant if the underlying stochastic processes have memory.

\subsection{Stochastic Process with Memory}
\label{subsec:corr_process}

Here we introduce temporal memory in the generation of content lifetimes and the user location, i.e., distance from the serving BS, which is  a more realistic model for a mobile user served by micro BSs. For the lifetime process, we assume that the content generator has two states, called ``short'' and ``long'' content states, respectively. When it is in the short content state, all the generated contents have an initial lifetime of $5$, whereas in the long content state, all the contents are generated with a lifetime of $15$. The content generator transitions from one state to the other randomly. We assume that, if it is in the short content state, it remains there with probability $p_1$, while if it is in the long content state it remains in taht state with probability $p_2$. 

We assume that at each time slot the user moves either towards or away from the serving BS. The distance from the serving base station at time slot $t+1$ is $d_{t+1} = d_t \pm \sigma$, where $\sigma$ is a positive constant, which describes how fast the user is moving. The corresponding user location model is a Markov process where the state transition probabilities are given by $P(d_{t+1} = d_t + \sigma) = p_u$ and $P(d_{t+1} = d_t - \sigma) = 1 - p_u$, $\forall t$. We further impose that the distance $d$ is bounded between $50\mathrm{m}$ and $250\mathrm{m}$, so only a single direction of movement is possible  on the boundary points.

Note that the threshold structure of the optimal policy detailed in Section \ref{sec:optimal} no longer holds in this model with memory; however, we can still evaluate the performances of the proposed caching schemes which exploit a threshold structure. Figure \ref{fig:corr_avg_cost} shows the performances of the LFA and LISO policies, both implemented with the LRM algorithm, for values of $p_1 \in \{0.1,0.5,0.9\}$ and varying $p_2$ from $0.1 ~\mbox{to}~0.9$. Note that, higher $p_1$ and $p_2$ values mean that the system is more likely to stay in the same state, and continue to generate contents with the same lifetime. The results are obtained for a cache capacity of $B=20$. For the user location process, we set $\sigma = 5$ and $p_u = 0.5$. We observe that the average energy cost increases with increasing $p_2$ and with decreasing $p_1$, as they both lead to the generation of more contents with lifetime $15$. We observe similar trends for the gain of LFA with respect to LISO; that is, the improvement with respect to LISO also increases with $p_2$ and decreases with $p_1$.

\begin{figure}[t]
\centering
\includegraphics[scale=0.45]{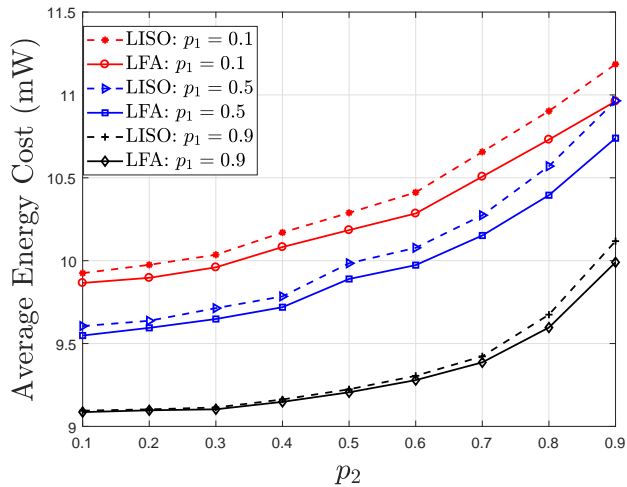}
\caption{Average energy cost vs. transition probabilities $p_1$ and $p_2$ for both LFA and LISO policies, with cache capacity $B=20$, $\sigma=5$ and $p_u=0.5$.}
\label{fig:corr_avg_cost}
\vspace{-0.5cm}
\end{figure}

We observe that the performance gain of LFA over LISO is more significant than the i.i.d. scenario. For similar system parameters in the i.i.d case; that is, for a cache capacity of $B=20$, and assuming that the LRM algorithm is used, LFA policy has a performance gain of approximately $0.75\%$ over LISO. However, when memory is introduced to the content lifetime and user location processes, the LFA policy can have a performance gain of approximately $2\%$. We note that, when the lifetime generation has memory, existing contents in the system provide more information about the future states; and hence, the LFA policy, which takes into account the remaining lifetimes of all the contents, provides larger gains.

\section{Conclusions}\label{s:conclusion}

We have considered the proactive caching problem in wireless networks with the aim of minimizing the long term average energy cost of delivering contents to the UE over a time-varying wireless link under random user accesses to the system, random content lifetime, and a time-varying library size. We have first showed the optimality of a threshold-based policy, which pushes contents to the cache (or may remove contents from the cache if it is full) depending on the relative value of the channel state with respect to preset threshold values that depend on the time elapsed since last user access and the remaining lifetimes of all the relevant contents in the system. Since this leads to a prohibitively large set of parameters to be optimized, we have proposed two suboptimal caching schemes, LISO and LFA, that are based on low-complexity parametrization of the system states and policy search techniques from reinforcement learning. We have further introduced two lower bounds on the performance, and through numerical simulations, we have showed that the two low-complexity proactive caching schemes perform close to optimal, with LFA performing better than LISO in general. Proactive caching under nonlinear cost functions, and in multi-user scenarios are currently being considered as interesting future extensions of this paper. 


\appendices

\section{Proof of Lemma \ref{lem:piecewise_constant}}
\label{app:proof-piecewise_constant}

We start the proof by showing that our MDP $(\S, \A, P, \mu)$ satisfies \eqref{eq:opt_diff_valuefunc} when $P^\pi$ is ergodic for any policy $\pi$ and $\Z$ is an interval. First note that Theorems 5.1--5.3 of \cite{Averageoptimality:survey} imply that for any MDP with a countable state space and whose action space is a compact metric space, there exists an optimal deterministic policy satisfying \eqref{eq:opt_diff_valuefunc}. Clearly, under our assumptions, $\S$ is countable. Furthermore, since $g \in \A$ is Borel-measurable, any limit point (under pointwise convergence) of a sequence of functions from $\A$ also belongs to $\A$ (i.e., it is a Borel-measurable function). On the other hand, the representation of policies with functions from $\A$ is not unique, since any two functions $g,g' \in \A$ such that $\PP{g(Z)= g'(Z)}=1$ represent the same policy (up to a zero-measure event), and this causes problems in establishing the compactness of $\A$. 

To alleviate this problem, for any $g \in \A$, define the equivalence class $\G_g=\{g'\in \G: \PP{g(Z)=g'(Z)}=1\}$, and let $\G=\{G_g: g \in \A\}$ denote the family of these classes. For any $G \in \G$, let $f_G \in G$ be a selected element of $G$. Then, since each function $f_G$ can take values only in the finite set $\A_{SI}$, with a slight modification to the proof of Theorem~3 in \cite{GyLi02}, one can show that the set $\bar{\A}=\{f_G: G \in \G\}$ is a compact metric space for the metric $\PP{g(Z) \neq g'(Z)}$. 


Consequently, the new MDP $(\S,\bar{\A},P,\mu)$ satisfies \eqref{eq:opt_diff_valuefunc}. Furthermore, it is easy to see that the new MDP is equivalent to the original one in the sense that their trajectories equal with probability one if any action $g \in \A$ in the original MDP is replaced with $f_{G_g}$ in the new one. Therefore, the original MDP also satisfies \eqref{eq:opt_diff_valuefunc}.

Using \eqref{eq:MDPtoMDP-SI_prob} and \eqref{eq:MDPtoMDP-SI_cost}, we can express \eqref{eq:opt_diff_valuefunc} as
\begin{equation}
V^{\pi^*}(s) = \min_{g \in \A}\left\{\EE{\mu_{SI}(s,g(Z),Z)- \rho^{\pi^*}+\sum_{s' \in \S} P(s'|s,g(Z)) V^{\pi^*}(s')}\right\},
\label{eq:MDPtoMDPSI-valuefunc}
\end{equation}
such that, if $g=\pi^*(s)$, then $g$ minimizes \eqref{eq:MDPtoMDPSI-valuefunc}

Since $g$ is a mapping from $\Z$, the above minimum can be realized by minimizing for each value of the side information $Z$ independently. Indeed, for any $s$, the minimum in \eqref{eq:MDPtoMDPSI-valuefunc} is achieved by any $g$ satisfying
\begin{equation}
g(z) \in \argmin_{a_{SI} \in \A_{SI,s}} \left\{\mu_{SI}(s,a_{SI},z) - \rho^{\pi^*} + \sum_{s' \in \S} P(s'|s,a_{SI}) V^{\pi^*}(s')\right\}.
\label{eq:MDPtoMDPSI-optpolicy}
\end{equation}

Since the right hand side of \eqref{eq:MDPtoMDPSI-optpolicy} is a minimum of finitely many linear functions, it follows that $g(z)$ can be chosen to be a piecewise constant function: a piecewise constant function over the interval $\Z$ is defined by an interval partition $\Z_1,\ldots,\Z_m$ of $\Z$ (for some $m$) and some actions $a_1,\ldots,a_m \in \A_{SI}$ such that $g(z)=a_i$ if $z \in \Z_i, i=1,\ldots,m$. Converting this policy back to the original MDP-SI problem finishes the proof.

\bibliographystyle{IEEEtran}
\bibliography{Journal_ref}
\end{document}